\documentclass[letterpaper, 10 pt, conference]{ieeeconf}  % Comment this line out
\IEEEoverridecommandlockouts                              % This command is only
\usepackage{graphicx} % for pdf, bitmapped graphics files
\usepackage{epsfig} % for postscript graphics files
\usepackage{amsmath} % assumes amsmath package installed
\usepackage{amssymb}  % assumes amsmath package installed
\usepackage{bbm}
\usepackage{empheq}
\usepackage{arydshln}
\usepackage{subfig}

\let\labelindent\relax
\usepackage{enumitem}
\usepackage{color}

\usepackage{amsthm}

\usepackage{blindtext}

\newcommand{\nn}{\nonumber}

\newcommand{\qij}{q_{ij}}
\newcommand{\qjl}{q_{j\ell}}

\newcommand{\Pro}{\mathbb{P}}
\newcommand{\E}{\mathbb{E}}

\newcommand{\prodb}[2]{\prod_{#1\sim #2}(1-\beta\mathbbm{1}_{X_#1(t)})}

\newcommand{\vp}{\vspace{3pt}}			% newcommands etc...

\usepackage{etoolbox}
\apptocmd{\thebibliography}{\scriptsize}{}{}

\title{\LARGE \bf
Improved Bounds on the Epidemic Threshold of Exact SIS Models on Complex Networks
}

\author{Navid Azizan Ruhi, Christos Thrampoulidis and Babak Hassibi% <-this % stops a space
\thanks{This work was supported in part by the National Science Foundation under grants CNS-0932428, CCF-1018927, CCF-1423663 and CCF-1409204, by a grant from Qualcomm Inc., by NASA’s Jet Propulsion Laboratory through the President and Director’s Fund, by King Abdulaziz University, and by King Abdullah University of Science and Technology.}% <-this % stops a space
\thanks{The authors are with the California Institute of Technology, Pasadena, CA 91125, USA
        {\tt\small \{azizan, cthrampo, hassibi\}@caltech.edu}}%
}

\begin{document}

\newtheorem{theorem}{Theorem}
\newtheorem{lemma}[theorem]{Lemma}
\newtheorem{proposition}[theorem]{Proposition}
\newtheorem{corollary}[theorem]{Corollary}

\maketitle
\thispagestyle{empty}
\pagestyle{empty}

%%%%%%%%%%%%%%%%%%%%%%%%%%%%%%%%%%%%%%%%%%%%%%%%%%%%%%%%%%%%%%%%%%%%%%%%%%%%%%%%
\begin{abstract}

The SIS (susceptible-infected-susceptible) epidemic model on an arbitrary network, without making approximations, is a $2^n$-state Markov chain with a unique absorbing state (the all-healthy state). This makes analysis of the SIS model and, in particular, determining the threshold of epidemic spread quite challenging. It has been shown that the exact marginal probabilities of infection can be upper bounded by an $n$-dimensional linear time-invariant system, a consequence of which is that the Markov chain is ``fast-mixing'' when the LTI system is stable, i.e. when $\frac{\beta}{\delta}<\frac{1}{\lambda_{\max}(A)}$ (where $\beta$ is the infection rate per link, $\delta$ is the recovery rate, and $\lambda_{\max}(A)$ is the largest eigenvalue of the network's adjacency matrix). This well-known threshold has been recently shown not to be tight in several cases, such as in a star network. In this paper, We provide tighter upper bounds on the exact marginal probabilities of infection, by also taking pairwise infection probabilities into account. Based on this improved bound, we derive tighter eigenvalue conditions that guarantee fast mixing (i.e., logarithmic mixing time) of the chain. We demonstrate the improvement of the threshold condition by comparing the new bound with the known one on various networks and epidemic parameters.

\end{abstract}

%%%%%%%%%%%%%%%%%%%%%%%%%%%%%%%%%%%%%%%%%%%%%%%%%%%%%%%%%%%%%%%%%%%%%%%%%%%%%%%%
%\section{Introduction}
\section{Introduction}

The mathematical modeling and analysis of epidemic spread is of great importance in understating dynamical processes over complex networks (e.g. social networks) and has attracted significant interest from different communities in recent years. The study of epidemics plays a key role in many areas beyond epidemiology \cite{bailey1975mathematical}, such as viral marketing \cite{phelps2004viral,richardson2002mining}, network security \cite{alpcan2010network,acemoglu2013network}, and information propagation \cite{jacquet2010information,cha2009measurement}. Although there is a huge body of work on epidemic models, classical ones mostly neglect the underlying network structure and assume a uniformly mixed population, which is obviously far from reality. However, in recent years more realistic networked models have been introduced, and many interesting results are now known \cite{nowzari2016analysis,pastor2015epidemic}.

In the simplest case (the binary-state or SIS model) each node is in one of two different states: susceptible (S) or infected (I). During any time interval, each susceptible (healthy) node has a chance of being independently infected by any of its infected neighbors (with probability $\beta$). Further, during any time interval, each infected node has a chance of recovering (with probability $\delta$) and becoming susceptible again. For a network with $n$ nodes, this yields a Markov chain with $2^n$ states, which is referred to as the exact or ``stochastic'' model. Since analyzing this model is quite challenging, most researchers have resorted to $n$-dimensional linear and nonlinear approximations (the most common being the “mean-field” approximation), which are sometimes called ``deterministic'' models. This paper focuses on improving known bounds on the exact model.

The spreading process can be considered either as a discrete-time Markov chain or a continuous-time one. Although the discrete-time model is sometimes argued to be more realistic \cite{arenas2010discrete,ahn2014random}, there is no fundamental difference between the two, and similar results have been shown for both. We focus on the discrete-time Markov chain here.
%Comparing the discrete-time Markov chain model to the continuous-time Markov chain model described in \cite{Ganeshnetworktopology}, continuous-time Markov chain model allows only one flip of each node's epidemic state at each moment. However, the discrete-time model allows change of epidemic states for more than one node at the same time. The reason being that change of epidemic state for two or more nodes can occur at same time interval, even though they do not happen at same moment. The transition matrix of the embedded Markov chain of continuous-time model has nonzero entries only where the Hamming distance of row coordinate and column coordinate is 1. In other words, the number of different digits for $X,Y \in \{0,1\}^n$ should be 1 in order for the entry of the $X$-th row and the $Y$-th column to be nonzero. However, the transition matrix of the discrete-time Markov chain model can have nonzero entries everywhere (except the row of the absorbing state). 

It is known that these epidemic models exhibit a phase transition behavior at a certain threshold \cite{castellano2010thresholds,barrat2008dynamical} , i.e., once the effective infection rate $\tau=\frac{\beta}{\delta}$ approaches a critical value $\tau_c$ \cite{nowzari2016analysis} the epidemic appears not to die out. We should remark that the Markov chain has a unique absorbing state, which is the all-healthy state, because once the system reaches this state it remains there forever since there are no infected nodes to propagate infections. This means that if we wait long enough the epidemic will eventually die out, which may seem to be odd at first. However, what this means is that the question of the epidemic dying out is not interesting; what is interesting is the question of {\em how long it takes for the epidemic to die out?} In particular, if the mixing time of the Markov chain is exponentially large one will not see it die out in any reasonable time. Therefore the right question to ask is what is the mixing time of the Markov chain (or, equivalently, its mean-time-to-absorption); it turns out that the threshold $\tau_c$ corresponds to the phase transition between ``slow mixing" (exponential time) and ``fast mixing" (logarithmic time) of the MC \cite{draief2010epidemics,van2013decay,van2014upper}.

The epidemic threshold (critical value) of general networks is still an open problem. However,  lower- and upper-bounds have been found using different techniques \cite{draief2010epidemics,van2014upper}. The best known \emph{lower-bound} is $1/\lambda_{\max}(A)$, i.e. the inverse of the leading eigenvalue of the adjacency matrix, which is derived by \emph{upper-bounding} the marginal probabilities of infection and using a linear dynamical system. In fact, this method relies on keeping track of $n$ variables which are upper bounds on the marginal probability of infection for any of the nodes. In this paper, we focus on improving this upper-bound on the infection probabilities and ultimately the lower-bound on the epidemic threshold. The key idea, is to maintain the ``pairwise'' probabilities of nodes' infections, in addition to the marginals. This comes at the cost of increased, yet still perfectly feasible, computation. There is a trade-off between the tightness of the bound and the complexity, and in theory if one takes into account all marginals, pairs, triples, and higher order terms, we get back to the original $2^n$-state Markov chain.

We first briefly review the known bound with marginals, and show a simple alternative approach for deriving it. We then move on to pairs and use the machinery developed in Section~\ref{sec:marginals} to derive tighter bounds on the probabilities and connect it with the mixing time of the Markov chain (Sections \ref{sec:p_ij} and \ref{sec:q_ij}). Finally, we demonstrate the improvement of the bounds through extensive simulations (Section~\ref{sec:simulations}), and conclude with future directions.

%%%%%%%%%%%%%%%%%%%%%%%%%%%%%%%%%%%%%%%%%%%%%%%%%%%%%%%%%%%%%%%%%%%%%%%%%%%%%%%%
%\section{The Markov Chain and Marginal Probabilities of Infection}\label{sec:marginals}
\section{The Markov Chain and Marginal Probabilities of Infection}\label{sec:marginals}

\begin{figure}[tpb]
  \centering
    \includegraphics[width=0.4\columnwidth]{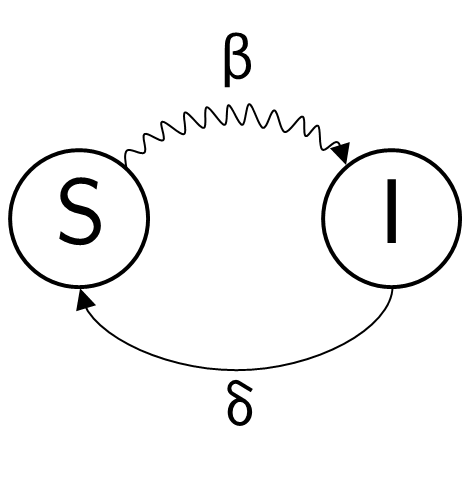}
      \caption{State diagram of a single node in the SIS model, and the transition rates. Wavy arrow represents exogenous (neighbor-based) transition. $\beta:$ probability of infection per infected link, $\delta:$ probability of recovery.}
      \label{fig}
\end{figure}

Let $G = (V,E)$ be an arbitrary connected undirected network with $n$ nodes, and with adjacency matrix $A$. Each node can be in a state of health, represented by ``0'', or a state of infection, represented by ``1''. The state of the entire network can be represented by a binary n-tuple $\xi(t)=(\xi_1(t), \cdots, \xi_n(t)) \in \{0,1 \}^n$, where each of the entries represents the state of a node at time $t$, i.e. $i$ is infected if $\xi_i(t)=1$ and it is healthy if $\xi_i(t)=0$.

Given the current state $\xi(t)$, the infection probability of each node in the next step is determined independently, and therefore the transition matrix $S$  of this Markov Chain has elements $S_{X,Y}=\mathbb{P}(\xi(t+1)=Y|\xi(t)=X)$ of the following form:
\begin{equation}
\mathbb{P}(\xi(t+1)=Y|\xi(t)=X) = \prod_{i=1}^n \mathbb{P}(\xi_i(t+1)=Y_i|\xi(t)=X) , \label{MC0}
\end{equation}
for any two state vectors $X,Y \in \{0,1\}^n$.

As mentioned before, a healthy node can receive infection from any of its infected neighbors independently with probability $\beta$ per infected link, and an infected node can recover from the disease with probability $\delta$. That is
\begin{align}
&\mathbb{P}(\xi_i(t+1)=Y_i|\xi(t)=X) \nonumber \\
&= \left\{
\begin{array}{rl}
(1-\beta)^{m_i} & \text{if } (X_i,Y_i)=(0,0),  |N_i \cap \mathbb{S}(X)| = m_i,\\
1- (1-\beta)^{m_i} & \text{if } (X_i,Y_i)=(0,1),  |N_i \cap \mathbb{S}(X)| = m_i,\\
\delta & \text{if } (X_i,Y_i)=(1,0),  |N_i \cap \mathbb{S}(X)| = m_i,\\
1-\delta & \text{if } (X_i,Y_i)=(1,1),  |N_i \cap \mathbb{S}(X)| = m_i,
\end{array} \right. \label{MC1}
\end{align}
where $\mathbb{S}(X)$ is the support of $X \in \{0,1\}^n$, i.e. $\mathbb{S}(X) = \{ i : X_i = 1 \}$, and $N_i$ is the set of neighbors of node $i$.

Eqs.~(\ref{MC0}, \ref{MC1}) completely define the $2^n\times 2^n$ transition matrix of the Markov chain, which determines the evolution of the $2^n$ states over time. Of course with this, we have the joint probability of all the nodes, and we can compute the probability of any desired combination by marginalizing out the rest. In particular, one can compute the probability of each node $i$ being infected at time $t+1$ (denoted by $p_i(t+1)$), which is a function of all joint probabilities of the states at time $t$. Since there are only $n$ such variables, the dimension would be significantly reduced if one could ``bound'' or ``approximate'' that function by something that includes marginals $p_i(t)$ only. This way we obtain a recursion which relates the marginals at time $t+1$ to those of time $t$, and indeed we have a system with only $n$ states rather that $2^n$ states. Approximations per se are not very interesting because they do not provide any \emph{guarantee} on the behavior of the exact Markov chain model. What is more important is whether one can obtain a bound on these true probabilities, which can guarantee for example fast extinction of disease. The most common upper-bound, which has been shown to be the tightest linear upper-bound with marginals only (using a linear programming technique) \cite{ahn2014mixing,azizan2015sirs} is:
\begin{equation} \label{eq:pi_classic}
p_i(t+1) \leq (1-\delta)p_i(t) + \beta \sum_{j \in N_i} p_j(t)
\end{equation}
for all $i=1,\dots,n$. Defining $p(t) = (p_1(t),\dots,p_n(t))^T$, this can be written in a matrix form as
\begin{equation}\label{eq:pi_bound}
p(t+1) \preceq M p(t) ,
\end{equation}
where
\begin{equation}
M=(1-\delta)I_n + \beta A.
\end{equation}

\subsection{An Alternative Bounding Technique}\label{sec:alt}
The derivation of \eqref{eq:pi_classic} in \cite{ahn2014mixing,azizan2015sirs} involves a linear programming technique. In this paper, we provide an alternative technique to bound the infection probabilities using indicator variables and conditional expectation, which is more intuitive and more direct. Importantly, as will be shown later, this technique can be used to obtain tighter bounds on the exact probabilities of infections using pairwise inflectional probabilities. Before that, it is instructive to derive  \eqref{eq:pi_classic} using this alternative approach.

Let $i\in V$. We start by conditioning on the state of the same node $i$ at time $t$, as follows:
\begin{align}
p_i(t+1)=&\mathbb{P}(X_i(t+1)=1 | X_i(t)=1) \mathbb{P}(X_i(t)=1)\notag\\ 
&+\mathbb{P}(X_i(t+1)=1 | X_i(t)=0) \mathbb{P}(X_i(t)=0)\nn.
\end{align}
The probability that an infected node remains infected is $1-\delta$, and the probability that a susceptible node does not receive infection from an infected neighbor is $1-\beta$. We denote $j$ neighbor of $i$ by $j\sim i$. The expression above can be written as
\begin{align}
p_i(t+1)&=(1-\delta)p_i(t)+\label{eq:up1} \\
&\mathbb{E}_{X_{-i}(t)|X_i(t)=0}\bigg[1-\prod_{j\sim i}(1-\beta\mathbbm{1}_{X_j(t)})\bigg]\mathbb{P}(X_i(t)=0)\nn .
\end{align}
The conditional expectation is on the joint probability of all nodes other than $i$ (denoted by $X_{-i}$) given node $i$ being healthy ($X_i=0$). Of note, this expression is still exact, and we have not done any approximation yet. It can be easily checked that $$\prod_{j\sim i}(1-\beta\mathbbm{1}_{X_j(t)})\geq 1 - \beta\sum_{j\sim i}\mathbbm{1}_{X_j(t)}.$$ Combining this with \eqref{eq:up1}, yields the desired upper bound
\begin{align}
p_i(t+1)&\leq (1-\delta)p_i(t)+\beta\sum_{j\sim i} \Pro(X_i(t)=0, X_j(t)=1) \label{eq:bound}\\&
\leq (1-\delta)p_i(t)+\beta\sum_{j\sim i}p_j(t)\notag.
\end{align}

\subsection{Connection to Mixing Time of the Markov Chain}\label{sec:connection}
Up to this point, we just talked about bounding the marginal probabilities of infection, and it is not clear how a bound on the marginal probabilities relates to the mixing time of the Markov chain. To establish this connection, let us start from the definition of mixing time \cite{levin2009markov}:
\begin{equation}
t_{mix}(\epsilon)=\min \{t: \sup_\mu \| \mu S^t - \pi \|_{TV} \leq \epsilon \} ,  \label{Eq:mixingdefn}
\end{equation}
where $\mu$ is any initial probability distribution defined on the state space and $\pi$ is the stationary distribution; $\|\mu-\mu' \|_{TV}$ is the total variation distance of any two probability measures $\mu$ and $\mu'$, and is defined by
\begin{equation}\nn
\| \mu -\mu' \|_{TV} = \frac{1}{2} \sum_x | \mu(x) - \mu'(x) | ,
\end{equation}
where $x$ is any possible state in the probability space. In fact $t_{mix}(\epsilon)$ is the minimum time instant for which the distance between the stationary distribution and the probability distribution at time $t$ from any initial distribution is smaller than or equal to $\epsilon$. Roughly speaking, the mixing time measures how fast the initial distribution converges to the limit distribution, which in our case means how quickly the epidemic dies out.

Since in the stationary distribution the all-healthy state has probability 1, it can be shown \cite{ahn2014mixing} that
\begin{align}
\sup_\mu \| \mu S^t - \pi \|_{TV} 
%&= 1-e_{\bar{1}} S^t e_{\bar{0}}\nn\\
&=\mathbb{P}\left( \substack{\text{some nodes are infected at time $t$} \mid \\\text{all nodes were infected at time $0$}} \right)
\end{align}
which highlights the fact that the worst initial distribution (i.e. the $\mu$ that maximizes above quantity) is the all-infected state.
Now for any $t< t_{mix}(\epsilon)$ we have
\begin{align}
\epsilon &<\mathbb{P}\left( \substack{\text{some nodes are infected at time $t$} \mid \\\text{all nodes were infected at time $0$}} \right)\nn\\
 &\leq \sum\limits_{i=1}^n \mathbb{P}\left( \substack{\text{node $i$ is infected at time $t$} \mid\qquad \\\text{all nodes were infected at time $0$}} \right)\nn\\
 &= 1_{n}^T p(t)\qquad \text{given that $p(0)=1_n$} , \label{eq:epsilon}
\end{align}
where we have used the union bound, and $1_n$ denotes the all-ones vector of size $n$.

Back to the upper-bound on the marginals \eqref{eq:pi_bound}, we get $1_n^Tp(t)\leq 1_n^TMp(t-1)$. Furthermore, since $M$ has non-negative entries (we write this as $M\geq 0$) we can ``propagate'' the bound to find that
\begin{equation}
1_n^Tp(t)\leq 1_n^TMp(t-1)\leq 1_n^T M^2p(t-2)\leq\dots\leq 1_n^T M^tp(0).\nn
\end{equation}
As a result, for any $t< t_{mix}(\epsilon)$
\begin{equation}
\epsilon<1_n^TM^t1_n\leq n (\rho(M))^t ,
\end{equation}
since $M$ is non-negative and symmetric, and $\lambda_{\max}(M)=\rho(M)$, where $\rho(M)$ is the spectral radius of $M$.

When $\rho(M)<1$ (or equivalently $1-\delta+\beta\lambda_{max}(A)<1$), it follows that $t< \frac{\log \frac{n}{\epsilon}}{-\log \rho(M)}$ for all $t< t_{mix}(\epsilon)$. This implies the well-known result that when $\beta/\delta<1/\lambda_{\max}(A)$ then $t_{mix}(\epsilon)\leq \frac{\log \frac{n}{\epsilon}}{-\log \rho(M)}=O(\log n)$.

We should note here that if $M$ was not symmetric (as we will encounter such instances in the next section), it can be shown by an appeal to the Lyapunov equation that if $\rho(M)<1$ then for all $t<t_{mix}(\epsilon)$ there exists $0<\eta<1$ such that  $\epsilon\leq \eta^tO(\mathrm{poly}(n)) $, from which it follows directly that the mixing time is logarithmic in $n$. To see that, note that $\rho(M)<1$ implies that there exists a positive definite matrix $P\succ0$ such that $M^TPM-P\prec0$. Letting $P^{1/2}$ denote the unique positive square root of $P$ and $N:=P^{1/2}MP^{-1/2}$, it follows easily that 
$N^TN \prec I_d $, or equivalently $\eta:=\|N\|_2<1$. (Here, $\|N\|_2$ denotes the spectral norm of $N$.) Defining $y:=P^{1/2}1_n$ and $x:=P^{-1/2}1_n$ we get
$
1_n^TM^t1_n=x^TN^ty\leq \|x\|_2\eta^t\|y\|_2 \leq n\eta^t \|P^{1/2}\|_2\|P^{-1/2}\|_2.
$

%%%%%%%%%%%%%%%%%%%%%%%%%%%%%%%%%%%%%%%%%%%%%%%%%%%%%%%%%%%%%%%%%%%%%%%%%%%%%%%%
%\section{Pairwise Probabilities ($p_{ij}$)}
%%%%%%%%%%%%%%%%%%%%%%%%%%%%%%%%%%%%%%%%%%%%%%%%%%%%%%%%%%%%%%%%%%%%%%%%%%%%%%%%
\section{Pairwise Probabilities ($p_{ij}$)}\label{sec:p_ij}

%%%%%%%%%%%%%%%%%%%%%%%%%%%%%%%%%%%%%%%%%%%%%%%%%%%%%%%%%%%%%%%%%%%%%%%%%%%%%%%%
%\subsection{A tighter bound on the $p_{i}$'s}
%
In section~\ref{sec:marginals} we showed how a bound on marginal probabilities of infection can be obtained, and how this bound translates to the threshold condition for fast mixing of the Markov chain. As mentioned before, the bound \eqref{eq:pi_bound} has been proved to be the tightest linear bound one can get with marginals; A natural idea to improve this bound is to go to higher order terms (i.e. pairs, triples, etc.). In principle, maintaining higher order terms is advantageous because it means keeping more information from the original chain, but of course at the cost of increased complexity. We define the pairwise probability of infection of two nodes, in addition to the marginals, as follows.
%In Section \ref{sec:alt} we obtained a bound on the probabilities $p_{i}(t+1)$ that only involved marginal probabilities at time $t$. Here, we tighten this bound by introducing pairwise infection probabilities.
For $(i,j)\in E$,
$$
p_{ij}(t) := \Pro( X_i(t) = 1 , X_j(t) = 1).
$$
Note that out of the $n\choose 2$ possible pairs of nodes, we only consider the ones that correspond to edges in the graph. Based on this definition, $\Pro( X_i(t) = 0 , X_j(t) = 1) = p_j(t) - p_{ij}(t)$ and it follows easily from \eqref{eq:bound} that
\begin{equation}
p_i(t+1) \leq (1-\delta)p_i(t) + \beta \sum_{j \sim i} p_j(t)  - \beta \sum_{j \sim i} p_{ij}(t) \label{eq:pij_pi}.
\end{equation}
Of course, this bound is at least as tight as the one in \eqref{eq:pi_classic}. Now, in order to strictly improve upon the latter, we need a lower bound on the pairwise infection probabilities at time $t+1$ in terms of marginals and pairwise probabilities at time $t$, which is derived next.

%%%%%%%%%%%%%%%%%%%%%%%%%%%%%%%%%%%%%%%%%%%%%%%%%%%%%%%%%%%%%%%%%%%%%%%%%%%%%%%%
\subsection{A Lower Bound on the $p_{ij}$'s}

To construct a lower bound on the pairwise marginal probabilities $p_{ij}(t+1)$, we use the same approach as was introduced in Section \ref{sec:alt}, but this time applied to pairwise infection probabilities. 

Let $(i,j)\in E$ and $t\geq 0$.
%In order to upper bound $q_{ij}(t+1)$ we condition on all four possible states of the pair $(i,j)$ at time $t$, i.e. 
We first expand $p_{ij}(t+1)$ as follows
\begin{align*}
\sum_{\substack{x\in\{0,1\}\\ y\in\{0,1\}}} \Pro(X_i(t+1)=1,X_j(t+1)=1 , X_i(t)=x,X_j(t)=y).
\end{align*}
For convenience, denote each one of the summands above by $s_{xy}$. Also, let $c_{xy} $ represent the corresponding conditional probability $\Pro(X_i(t+1)=1,X_j(t+1)=1 | X_i(t)=x,X_j(t)=y).$
In what follows, we lower bound each one of $s_{xy}$'s.  We write $\E_{xy}$ for the conditional expectation $\E_{X_{-i,-j}(t) | \{X_i(t)=x,X_j(t)=y\}}$.

\vp
\noindent$\bullet ~$\underline{(x=0,y=0):} Trivially, $s_{00}\geq 0$.

\vp
\noindent$\bullet ~$\underline{(x=0,y=1):} As before, the probability of not getting infected from each infected neighbor is $(1-\beta)$, and the probability that an infected node remains infected is $1-\delta$. Therefore
\begin{align}
c_{01}  = \E_{01}\big[ {(1-\prodb{k}{i})} (1-\delta) \big] \nn.
\end{align}
Since $j\sim i$ and $0\leq\beta\leq 1$, it follows that $\prodb{k}{i}\leq (1-\beta\mathbbm{1}_{X_j(t)})$. Hence $c_{01}\geq \beta(1-\delta)\E_{01} \mathbbm{1}_{X_j(t)} $, which eventually gives
\begin{align}s_{01} &\geq \beta(1-\delta) \Pro(X_j(t)=1,X_i(t)=0) \nn\\
&\qquad\qquad\geq  \beta(1-\delta) p_j(t) -   \beta(1-\delta) p_{ij}(t).\nn\end{align}

\vp
\noindent$\bullet ~$\underline{(x=1,y=0):} By symmetry, the exact same argument as above implies $$s_{10} \geq  \beta(1-\delta) p_i(t) -   \beta(1-\delta) p_{ij}(t).$$

\vp
\noindent$\bullet ~$\underline{(x=1,y=1):} Clearly
$c_{11}  = (1-\delta)^2, \nn$
which gives 
$$s_{11} \geq (1-\delta)^2p_{ij}(t).$$

Adding up all the above terms yields the following lower bound for all $(i,j)\in E$ and $t\geq 0$:
\begin{align}
p_{ij}(t+1) &\geq (1-\delta)\beta(p_i(t) + p_j(t)) \notag\\
&\qquad\quad+ (1-\delta)(1-\delta-2\beta)p_{ij}(t) \label{eq:pij_pij}. 
\end{align}

%%%%%%%%%%%%%%%%%%%%%%%%%%%%%%%%%%%%%%%%%%%%%%%%%%%%%%%%%%%%%%%%%%%%%%%%%%%%%%%%
\subsection{Back to the Mixing Time}
In order to express Eqs. \eqref{eq:pij_pi} and \eqref{eq:pij_pij} for all $i$ and $j$'s together in a matrix form, recall the definition $p(t) = (p_1(t),\dots,p_n(t))^T$. Further, let us define $p_E(t)\in \mathbb{R}^{|E|}$ as the vector of pairwise infection probabilities, i.e.,  $ p_E(t)= \mathrm{vec}(p_{ij}(t): (i,j)\in E)$. Note that $p_{ij}(t)=p_{ji}(t)$, so for each edge we only keep track of one of the two terms. Now we can write  Eqs. \eqref{eq:pij_pi}, \eqref{eq:pij_pij} as
\begin{equation}
 \begin{bmatrix} p(t+1) \\ -p_E(t+1) \end{bmatrix}
\preceq M' 
 \begin{bmatrix} p(t) \\ -p_E(t) \end{bmatrix} . \label{eq:pij_bound}
\end{equation}
The matrix $M'$, after a little bit of thought, can be expressed in the following way.
\begin{align}
M'=\begin{bmatrix}
(1-\delta)I_n + \beta A & \beta B  \\
-(1-\delta)\beta B^T & (1-\delta)(1-\delta-2\beta) I_{|E|}
\end{bmatrix} ,
\end{align}
where $B\in\mathbb{R}^{|V|\times |E|}$ happens to be the incidence matrix of $G$, which is formally defined as
$$
B_{i,e} = \begin{cases} 1 & \text{if $i$ is an endpoint of $e$},
\\ 0 & \text{otherwise},
\end{cases}
$$
for all $i\in V$ and $e\in E$. 
%It turns out that the $B$ matrix is nothing but the incidence matrix of the graph.

By accounting for pairwise infection probabilities, the bound derived in \eqref{eq:pij_bound} is tighter when compared to the one in \eqref{eq:pi_bound}. In order 
to connect this to the the mixing time of the underlying Markov chain, observe that
%
%However, it is not yet clear how this bound can be useful in tightening the bound on 
%the mixing time of the underlying Markov chain. The argument follows along the lines of Section \ref{sec:connection} as explained next. 
%
%Observe that,
\begin{align}
1_n^Tp(t) = \begin{bmatrix} 1_n^T & 0_{|E|}^T \end{bmatrix} \begin{bmatrix} p(t) \\ -p_E(t) \end{bmatrix}.\nn
\end{align}
Applying \eqref{eq:pij_bound} to this gives,
\begin{align}
1_n^Tp(t) \leq \begin{bmatrix} 1_n^T & 0_{|E|}^T \end{bmatrix} M' \begin{bmatrix} p(t-1) \\ -p_E(t-1) \end{bmatrix}. \label{eq:pij_t-1}
\end{align}
This step is possible because the entries of the vector  $\begin{bmatrix} 1_n^T & 0_{|E|}^T \end{bmatrix}$ are all non-negative, which guarantees that the signs of all the $n+|E|$ inequalities in \eqref{eq:pij_bound} are preserved. With this note, it becomes clear that in order to be able to propagate the bounds for the remaining time instances $t-2,t-3,\ldots,0$, a sufficient condition would be
\begin{align}\label{eq:pij_condition}
\begin{bmatrix} 1_n^T & 0_{|E|}^T \end{bmatrix} (M')^t \geq 0, \quad\text{for all } t\geq 1.
\end{align}

Provided that \eqref{eq:pij_condition} holds, we can continue with the sequence of bounds after \eqref{eq:pij_t-1}, which results in
\begin{align}
1_n^Tp(t) \leq \begin{bmatrix} 1_n^T & 0_{|E|}^T \end{bmatrix} (M')^t \begin{bmatrix} p(0) \\ -p_E(0) \end{bmatrix}. \label{eq:pij_0}
\end{align}
Subsequently, the same argument as in \ref{sec:connection} concludes the following result.
\begin{theorem}\label{thm:p_ij}
Assume that \eqref{eq:pij_condition} holds. If $\rho(M')<1$, then the mixing time of the Markov chain whose transition matrix $S$ is described by Eqs. \eqref{MC0} and \eqref{MC1} is $O(\log n)$.
\end{theorem}
From the old bound it was known that when $\rho(M)<1$ then the Markov chain is fast-mixing. Now in addition to that, this theorem states that also when $\rho(M')<1$ the Markov chain mixes fast. Of course, this is informative only when there is a case where $\rho(M)>1$ but $\rho(M')<1$. As it will be shown in the Section~\ref{sec:simulations} this is indeed the case.

Note that in the proof of Theorem~\ref{thm:p_ij} we used the assumption that \eqref{eq:pij_condition} holds. As it will be shown in the simulations section, in many cases this is a reasonable assumption. However, when the assumption does not hold we cannot appeal to this theorem. For this reason we propose another bound using an alternative pairwise probability, which does not require such a condition.

%%%%%%%%%%%%%%%%%%%%%%%%%%%%%%%%%%%%%%%%%%%%%%%%%%%%%%%%%%%%%%%%%%%%%%%%%%%%%%%%
%\section{An Alternative Pairwise Probability ($q_{ij}$)}
\section{An Alternative Pairwise Probability ($q_{ij}$)}\label{sec:q_ij}

As it was discussed above, when the assumption \eqref{eq:pij_condition} does not hold, we seek an alternative bound. Let us define
$$
q_{ij}(t) := \Pro( X_i(t) = 0 , X_j(t) = 1).
$$

We can use the same approach as before to obtain bounds for $p_i$, $q_{ij}$'s. Intuitively, lower bounding $p_{ij}(t+1)$ is equivalent with upper bounding $q_{ij}(t+1)$, and it turns out that it is what we need. The next lemma summarizes the bounds on these probabilities.
\begin{lemma} For all $i,j\in V$, $(i,j)\in E$ and $t\geq 0$, it holds that
\begin{align}
p_i(t+1) &\leq  (1-\delta) p_i(t) + \beta \sum_{\ell\sim i} q_{i\ell}(t) \label{eq:qij_pi} \\
q_{ij}(t+1) & \leq \delta(1-\delta) p_j(t) + (1-\delta)(1-\delta-\beta)q_{ij}(t) \nn
\\&\qquad\qquad+ \beta\delta q_{ji}(t) + \beta(1+\delta)\sum_{\substack{\ell\sim j\\ \ell\neq i}}q_{j\ell}(t) \label{eq:qij_qij}
\end{align}
\end{lemma}

\begin{proof}
Observe that \eqref{eq:qij_pi} is nothing but \eqref{eq:pij_pi} expressed in terms of $q_{ij}$'s. 
Now let $(i,j)\in E$. 
%In order to upper bound $q_{ij}(t+1)$ we condition on all four possible states of the pair $(i,j)$ at time $t$, i.e. 
We first expand $q_{ij}(t+1)$ as follows
\begin{align*}
\sum_{\substack{x\in\{0,1\}\\ y\in\{0,1\}}} \Pro(X_i(t+1)=0,X_j(t+1)=1 , X_i(t)=x,X_j(t)=y).
\end{align*}
For convenience, denote each one of the summands above by $s_{xy}$. Also, let $c_{xy} $ denote the corresponding conditional probabilities $\Pro(X_i(t+1)=0,X_j(t+1)=1 | X_i(t)=x,X_j(t)=y).$
In what follows, we upper bound each one of the $s_{xy}$'s; this will immediately yield \eqref{eq:qij_qij}.  All expectations below are conditional on the events $\{X_i(t)=x,X_j(t)=y\}$ which is omitted for the sake of convenience.

\vp
\noindent$\bullet ~$\underline{(x=0,y=0):} Using the fact that
\begin{align}
c_{00}  = \E\big[ \underbrace{(\prodb{k}{i})}_{\leq 1} \underbrace{( 1 - \prodb{\ell}{j} )}_{\leq \beta\sum_{\ell\sim j}\mathbbm{1}_{X_\ell}} \big] , \nn
\end{align}
we find that $s_{00} \leq \beta\sum_{\ell\sim j}\Pro(X_i(t)=0,X_j(t)=0,X_\ell(t)=1)\leq \beta\sum_{\substack{\ell\sim j \\ \ell\neq i}}\qjl(t)$.

\vp
\noindent$\bullet ~$\underline{(x=0,y=1):} From
\begin{align}
c_{01}  = \E\big[ {(\prodb{k}{i})} (1-\delta) \big] \leq 1-\beta\mathbbm{1}_{X_j} , \nn
\end{align}
it follows that $s_{01} \leq (1-\beta)(1-\delta) \qij(t)$.

\vp
\noindent$\bullet ~$\underline{(x=1,y=0):} Using the fact that
\begin{align}
c_{10}  = \E\big[ \delta {( 1 - \prodb{\ell}{j} )} \big]\leq \beta\delta\sum_{\ell\sim j}\mathbbm{1}_{X_\ell} , \nn
\end{align}
we find that $s_{10} \leq \beta\delta\sum_{\ell\sim j}\Pro(X_i(t)=1,X_j(t)=0,X_\ell(t)=1)\leq \beta\delta\sum_{\substack{\ell\sim j}}\qjl(t)$.

\vp
\noindent$\bullet ~$\underline{(x=1,y=1):} Using the fact that
$c_{11}  = \delta(1-\delta), \nn$
we find that $s_{11} = \delta(1-\delta)\Pro(X_i(t)=1,X_j(t)=1)=\delta(1-\delta)(p_j-q_{ij})$.

%\vp 
%This shows \eqref{eq:qij_qij}.
\end{proof}
\vspace{-10pt}
Similar as before, by defining a vector $q_E(t)\in \mathbb{R}^{2|E|}$ as $q_E(t)= \mathrm{vec}(q_{ij}(t): (i,j)\in E)$, we can express \eqref{eq:qij_pi}, \eqref{eq:qij_qij} as
\begin{equation}
 \begin{bmatrix} p(t+1) \\ q_E(t+1) \end{bmatrix}
\preceq M'' 
 \begin{bmatrix} p(t) \\ q_E(t) \end{bmatrix}, \label{eq:qij_bound}
\end{equation}
for some appropriately defined square matrix $M''$ of size ${n+2|E|}$.
%Similar as before, by defining a vector $q(t)\in \mathbb{R}^{n^2-n}$ as $q(t)= \mathrm{vec}(q_{ij}(t): i\neq j)$, we can express Eqs. \eqref{eq:qij_pi}, \eqref{eq:qij_qij} as
%\begin{equation}
% \begin{bmatrix} p(t+1) \\ q(t+1) \end{bmatrix}
%\preceq M'' 
% \begin{bmatrix} p(t) \\ q(t) \end{bmatrix} . \label{eq:pij_bound}
%\end{equation}
It is easy to see that if $1-\delta-\beta\geq 0$, then $M''\geq 0$, i.e. all entries of M are nonnegative. In particular, this implies that $M''$ satisfies \eqref{eq:pij_condition} and it only takes repeating the same argument as in \eqref{eq:pij_0} to conclude with the following theorem.

%Now the advantage of this is that all entries of $M''$ are positive and the bound can be easily propagated over time.
\begin{theorem}\label{thm:q_ij}
If $\rho(M'')<1$ and $1-\delta-\beta\geq0$, then the mixing time of the Markov chain whose transition matrix $S$ is described by Eqs. \eqref{MC0} and \eqref{MC1} is $O(\log n)$.
\end{theorem}
%\proof{the proof is similar to that of Theorem~\ref{thm:p_ij} and is omitted for brevity.}

%\addtolength{\textheight}{-3cm}   % This command serves to balance the column lengths

%%%%%%%%%%%%%%%%%%%%%%%%%%%%%%%%%%%%%%%%%%%%%%%%%%%%%%%%%%%%%%%%%%%%%%%%%%%%%%%%
%\section{Experimental Results}
\section{Experimental Results}\label{sec:simulations}

\begin{table}
\renewcommand{\arraystretch}{1.5}
\caption{Performance of the proposed bounds $M'$ and $M''$ in comparison with the old bound $M$. Boldface values show an improvement over the old bound. The signs next to $\rho(M')$ indicate whether the non-negativity condition (required for the proof) holds.}
\vspace{-23pt}
\label{table}
\begin{center}
\begin{tabular}{c c|c|c c|c}
 & & \begin{tabular}[x]{@{}c@{}}$\rho(M)=$\\$1+\epsilon$\end{tabular} &$\rho(M')$ & &$\rho(M'')$\\
\hline
Star & $\delta$=0.750, $\beta$=0.078&1.030 & \textbf{0.903}&+\\
 & $\delta$=0.500, $\beta$=0.053&1.030 & \textbf{0.828}&+\\
 & $\delta$=0.250, $\beta$=0.028&1.030 & 0.840&-&\textbf{0.968}\\
\hline
Cycle & $\delta$=0.750, $\beta$=0.390&1.030 & \textbf{0.817}&+\\
 & $\delta$=0.500, $\beta$=0.265&1.030 & 0.720&-&\textbf{0.945}\\
 & $\delta$=0.250, $\beta$=0.140&1.030 & 0.882&-&\textbf{0.942}\\
\hline
Star-line & $\delta$=0.750, $\beta$=0.174&1.030 & \textbf{0.872}&+\\
 & $\delta$=0.500, $\beta$=0.118&1.030 & 0.693&-&\textbf{0.958}\\
 & $\delta$=0.250, $\beta$=0.063&1.030 & 0.856&-&\textbf{0.955}\\
\hline
Clique & $\delta$=0.750, $\beta$=0.008&1.003 & \textbf{0.999}&+\\
 & $\delta$=0.500, $\beta$=0.005&1.003 & \textbf{0.998}&+&\\
 & $\delta$=0.250, $\beta$=0.003&1.003 & \textbf{0.999}&+&\\
\hline
Erd\H{o}s-R\'enyi & $\delta$=0.750, $\beta$=0.070&1.030 & \textbf{0.993}&+\\
 & $\delta$=0.500, $\beta$=0.048&1.030 & \textbf{0.977}&+\\
 & $\delta$=0.250, $\beta$=0.026&1.030 & \textbf{0.984}&+\\
\hline
Watts-Strogatz & $\delta$=0.750, $\beta$=0.077&1.030 & \textbf{0.991}&+\\
 & $\delta$=0.500, $\beta$=0.053&1.030 & \textbf{0.974}&+\\
 & $\delta$=0.250, $\beta$=0.028&1.030 & \textbf{0.982}&+\\

\end{tabular}
\end{center}
\vspace{-\baselineskip}
\end{table}

In this section, we demonstrate the performance of the proposed bounds by evaluating them on a variety of networks such as clique, Erd\H{o}s-R\'enyi, Watts-Strogatz, star graph, line graph, cycle, and star-line graph, with various parameters $\beta$ and $\delta$. As mentioned before, in order for any of the two threshold conditions proposed in sections \ref{sec:p_ij} and \ref{sec:q_ij} to be an improvement, we need to check if there are cases where the spectral radius of $M'$ or $M''$ is \emph{less than} 1, while the spectral radius of $M$ is \emph{greater than} 1 (or equivalently $\frac{\beta\lambda_{\max}(A)}{\delta}>1$). Indeed extensive simulations on our first bound ($M'$) suggest that not only are there such cases, but interestingly always $\rho(M')\leq\rho(M)$.

In order to compare $M'$ with $M$, we set $\rho(M)=1+\epsilon>1$ for some small value of $\epsilon$, and observe the value of $\rho(M')$. Table~\ref{table} lists the values of spectral radii for the three matrices. The positive sign next to $\rho(M')$ indicates that the non-negativity condition \eqref{eq:pij_condition} holds. For the cases that the condition holds (+), we can conclude that $M'$ has clearly improved. For the cases where the condition does not hold (-) we evaluate the second proposed bound $M''$, which again shows clear improvement over the old bound.

In order to demonstrate how tight the new condition is, Fig.~\ref{plot} plots the evolution of the epidemic over a star graph, for which the old bound is known not to be tight. The parameters in the two cases are $\delta=0.3, \beta=0.0130$, and  $\delta=0.3, \beta=0.0157$. It can be seen that while the value of $\frac{\beta\lambda_{\max}(A)}{\delta}$ is not informative (it is $1.93>1$ for the first case, and $2.33>1$ for the second one), the $\rho(M')$ condition is quite tight.

\begin{figure}[tp]
  \centering
    \includegraphics[width=0.9\columnwidth]{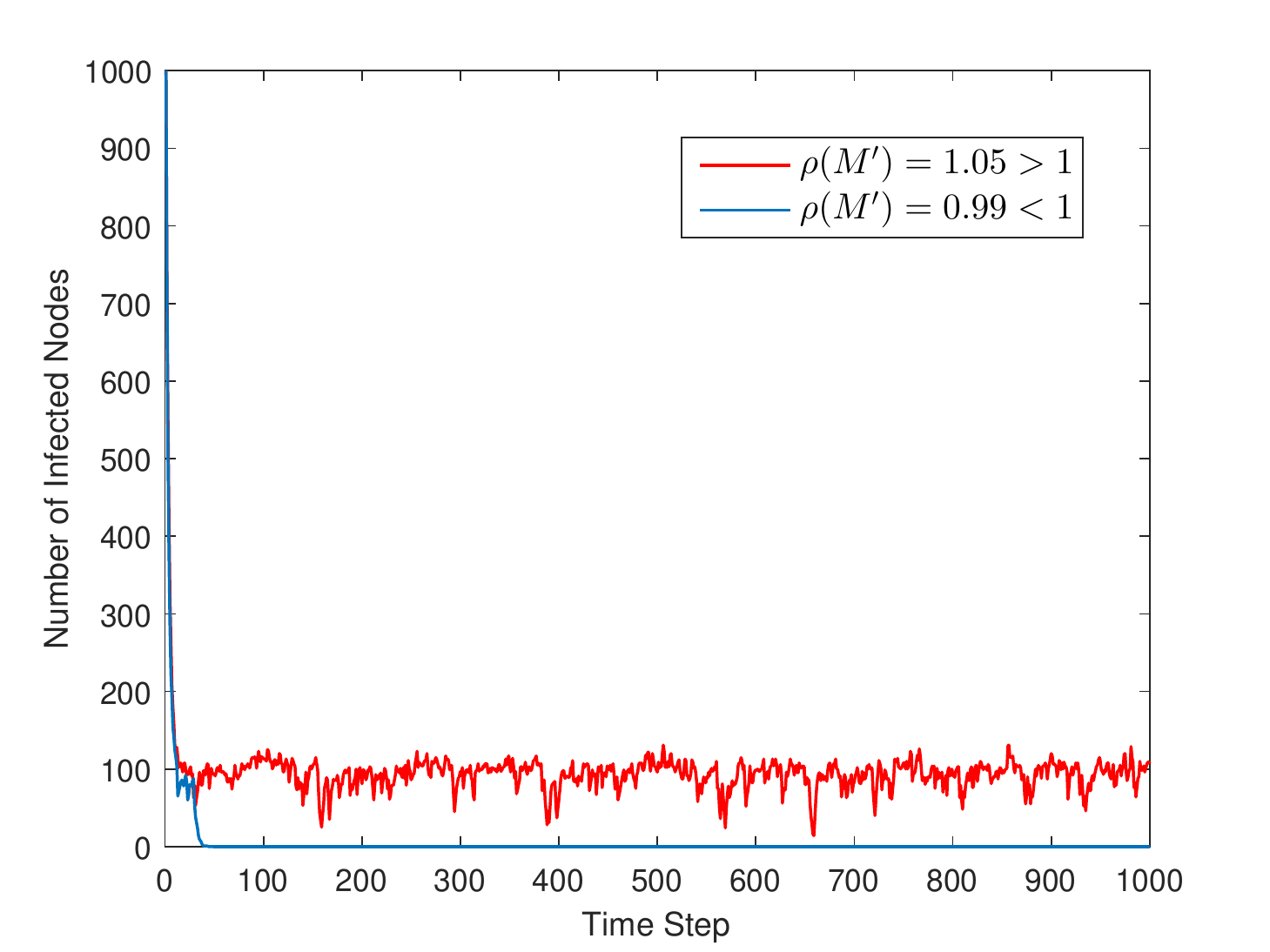}
      \caption{Evolution of the SIS epidemic over a star graph with $n=2000$ nodes, with two values of $\rho(M')$ below and above 1. When $\rho(M')=0.99<1$ we observe fast extinction of the epidemic (blue curve). It seems that the condition is also very tight since for $\rho(M')=1.05>1$ the epidemic does not die out (red curve).}
      \label{plot}
\vspace{-\baselineskip}
\end{figure}

%%%%%%%%%%%%%%%%%%%%%%%%%%%%%%%%%%%%%%%%%%%%%%%%%%%%%%%%%%%%%%%%%%%%%%%%%%%%%%%%
\vspace{-8pt}
\section{Conclusion and Future Work}
In this paper, we first proposed a simple technique using conditional expectations to systematically construct bounds on the exact probabilities of infection, up to any desired order. Using this approach, we showed that keeping higher order terms (such as pairs) indeed helps in obtaining tighter bounds; specifically we derived a bound composed of both marginals and pairwise probabilities which has improved over the well-known bounds. Based on this new bound, we provided a new condition for fast mixing of the Markov chain to the all-healthy state, which through extensive simulations was shown to be tighter than the so-called $\frac{\beta\lambda_{\max}(A)}{\delta}<1$ condition.

Clearly one possible extension of this work would be to construct a bound consisting of marginals, pairs, and triples, which in theory should result in an improvement. In fact, keeping all higher order terms eventually takes us back to the $2^n$-state Markov chain. Therefore there is a trade-off between the complexity and the accuracy of the model. We should however note that going to triples may still be tractable, and one advantage of that would be not only to gain by improving the bound on the probability (same as here) but also to get an improvement in the fast-mixing condition of the chain in the sense that the bound \eqref{eq:epsilon} can be replaced by
$\epsilon<\sum_i p_i -\sum_{i,j}p_{ij}+\sum_{i,j,k}p_{ijk}$, which naturally includes all terms rather than just the marginals.
As a last comment, based on the simulations we conjecture that condition \eqref{eq:pij_condition} may be relaxed, and other future work may concern its proof.

%%%%%%%%%%%%%%%%%%%%%%%%%%%%%%%%%%%%%%%%%%%%%%%%%%%%%%%%%%%%%%%%%%%%%%%%%%%%%%%%
\vspace{-4pt}
\section{Acknowledgment}
\small{
The authors would like to thank Ahmed Douik, Anatoly Khina and Ehsan Abbasi for insightful discussions on the subject.
}
%%%%%%%%%%%%%%%%%%%%%%%%%%%%%%%%%%%%%%%%%%%%%%%%%%%%%%%%%%%%%%%%%%%%%%%%%%%%%%%%
%\let\secfnt\undefined
%\newfont{\secfnt}{ptmb8t at 10pt}
\vspace{-4pt}
\bibliographystyle{IEEEtran}
\bibliography{IEEEabrv,references}

\begin{thebibliography}{10}
\providecommand{\url}[1]{#1}
\csname url@samestyle\endcsname
\providecommand{\newblock}{\relax}
\providecommand{\bibinfo}[2]{#2}
\providecommand{\BIBentrySTDinterwordspacing}{\spaceskip=0pt\relax}
\providecommand{\BIBentryALTinterwordstretchfactor}{4}
\providecommand{\BIBentryALTinterwordspacing}{\spaceskip=\fontdimen2\font plus
\BIBentryALTinterwordstretchfactor\fontdimen3\font minus
  \fontdimen4\font\relax}
\providecommand{\BIBforeignlanguage}[2]{{%
\expandafter\ifx\csname l@#1\endcsname\relax
\typeout{** WARNING: IEEEtran.bst: No hyphenation pattern has been}%
\typeout{** loaded for the language `#1'. Using the pattern for}%
\typeout{** the default language instead.}%
\else
\language=\csname l@#1\endcsname
\fi
#2}}
\providecommand{\BIBdecl}{\relax}
\BIBdecl

\bibitem{bailey1975mathematical}
N.~T. Bailey \emph{et~al.}, \emph{The mathematical theory of infectious
  diseases and its applications}.\hskip 1em plus 0.5em minus 0.4em\relax
  Charles Griffin \& Company Ltd, 5a Crendon Street, High Wycombe, Bucks HP13
  6LE., 1975.

\bibitem{phelps2004viral}
J.~E. Phelps, R.~Lewis, L.~Mobilio, D.~Perry, and N.~Raman, ``Viral marketing
  or electronic word-of-mouth advertising: Examining consumer responses and
  motivations to pass along email,'' \emph{Journal of advertising research},
  vol.~44, no.~04, pp. 333--348, 2004.

\bibitem{richardson2002mining}
M.~Richardson and P.~Domingos, ``Mining knowledge-sharing sites for viral
  marketing,'' in \emph{Proceedings of the eighth ACM SIGKDD international
  conference on Knowledge discovery and data mining}.\hskip 1em plus 0.5em
  minus 0.4em\relax ACM, 2002, pp. 61--70.

\bibitem{alpcan2010network}
T.~Alpcan and T.~Ba{\c{s}}ar, \emph{Network security: A decision and
  game-theoretic approach}.\hskip 1em plus 0.5em minus 0.4em\relax Cambridge
  University Press, 2010.

\bibitem{acemoglu2013network}
D.~Acemoglu, A.~Malekian, and A.~Ozdaglar, ``Network security and contagion,''
  National Bureau of Economic Research, Tech. Rep., 2013.

\bibitem{jacquet2010information}
P.~Jacquet, B.~Mans, and G.~Rodolakis, ``Information propagation speed in
  mobile and delay tolerant networks,'' \emph{Information Theory, IEEE
  Transactions on}, vol.~56, no.~10, pp. 5001--5015, 2010.

\bibitem{cha2009measurement}
M.~Cha, A.~Mislove, and K.~P. Gummadi, ``A measurement-driven analysis of
  information propagation in the flickr social network,'' in \emph{Proceedings
  of the 18th international conference on World wide web}.\hskip 1em plus 0.5em
  minus 0.4em\relax ACM, 2009, pp. 721--730.

\bibitem{nowzari2016analysis}
C.~Nowzari, V.~M. Preciado, and G.~J. Pappas, ``Analysis and control of
  epidemics: A survey of spreading processes on complex networks,''
  \emph{Control Systems, IEEE}, vol.~36, no.~1, pp. 26--46, 2016.

\bibitem{pastor2015epidemic}
R.~Pastor-Satorras, C.~Castellano, P.~Van~Mieghem, and A.~Vespignani,
  ``Epidemic processes in complex networks,'' \emph{Reviews of modern physics},
  vol.~87, no.~3, p. 925, 2015.

\bibitem{arenas2010discrete}
A.~Arenas, J.~Borge-Holthoefer, S.~Meloni, Y.~Moreno \emph{et~al.},
  ``Discrete-time markov chain approach to contact-based disease spreading in
  complex networks,'' \emph{EPL (Europhysics Letters)}, vol.~89, no.~3, p.
  38009, 2010.

\bibitem{ahn2014random}
H.~J. Ahn, ``Random propagation in complex systems: nonlinear matrix recursions
  and epidemic spread,'' Ph.D. dissertation, California Institute of
  Technology, 2014.

\bibitem{castellano2010thresholds}
C.~Castellano and R.~Pastor-Satorras, ``Thresholds for epidemic spreading in
  networks,'' \emph{Physical review letters}, vol. 105, no.~21, p. 218701,
  2010.

\bibitem{barrat2008dynamical}
A.~Barrat, M.~Barthelemy, and A.~Vespignani, \emph{Dynamical processes on
  complex networks}.\hskip 1em plus 0.5em minus 0.4em\relax Cambridge
  University Press, 2008.

\bibitem{draief2010epidemics}
M.~Draief and L.~Massouli, \emph{Epidemics and rumours in complex
  networks}.\hskip 1em plus 0.5em minus 0.4em\relax Cambridge University Press,
  2010.

\bibitem{van2013decay}
P.~Van~Mieghem, ``Decay towards the overall-healthy state in {SIS} epidemics on
  networks,'' \emph{arXiv preprint arXiv:1310.3980}, 2013.

\bibitem{van2014upper}
P.~Van~Mieghemy, F.~D. Sahnehz, and C.~Scoglioz, ``An upper bound for the
  epidemic threshold in exact {Markovian} {SIR} and {SIS} epidemics on
  networks,'' in \emph{Decision and Control (CDC), 2014 IEEE 53rd Annual
  Conference on}.\hskip 1em plus 0.5em minus 0.4em\relax IEEE, 2014, pp.
  6228--6233.

\bibitem{ahn2014mixing}
H.~J. Ahn and B.~Hassibi, ``On the mixing time of the {SIS} markov chain model
  for epidemic spread,'' in \emph{Decision and Control (CDC), 2014 IEEE 53rd
  Annual Conference on}.\hskip 1em plus 0.5em minus 0.4em\relax IEEE, 2014.

\bibitem{azizan2015sirs}
N.~A. Ruhi and B.~Hassibi, ``{SIRS} epidemics on complex networks: Concurrence
  of exact markov chain and approximated models,'' in \emph{2015 54th IEEE
  Conference on Decision and Control (CDC)}, Dec 2015, pp. 2919--2926.

\bibitem{levin2009markov}
D.~A. Levin, Y.~Peres, and E.~L. Wilmer, \emph{Markov chains and mixing
  times}.\hskip 1em plus 0.5em minus 0.4em\relax American Mathematical Soc.,
  2009.

\end{thebibliography}

\end{document}